\documentclass[a4paper,runningheads,envcountsect, envcountsame, orivec]{llncs}

\usepackage[T1]{fontenc}
\usepackage{lmodern}
\usepackage{amsmath,amssymb}
\usepackage{microtype}
\usepackage[hidelinks]{hyperref}

\hypersetup{
  pdfauthor={Tony Huynh, Eun Jung Kim, Sang-il Oum, Roohani Sharma, and Marek Sokolowski},
  pdftitle={Multiway Partition is FPT for Connectivity Functions}
}

\newcommand{\abs}[1]{\lvert #1\rvert}
\newcommand{\OPT}{\operatorname{OPT}}
\newcommand{\SFM}{\operatorname{SFM}}

\newcommand{\1}{\mathbf 1}
\newcommand{\co}[1]{\overline{#1}}

\title{Multiway Partition is FPT for Connectivity Functions}
\titlerunning{Multiway Partition is FPT for Connectivity Functions}

\author{Tony Huynh\inst{1} \and
Eun Jung Kim\inst{1,2} \and
Sang-il Oum\inst{1,3} \and
Roohani Sharma\inst{1} \and
Marek Soko\l{}owski\inst{4}}
\authorrunning{T. Huynh et al.}
\institute{Discrete Mathematics Group, Institute for Basic Science (IBS),
Daejeon, South Korea\\
\email{tony@ibs.re.kr, eunjungkim78@gmail.com, sangil@ibs.re.kr,
roohani@ibs.re.kr}
\and
School of Computing, KAIST, Daejeon, South Korea
\and
Department of Mathematical Sciences, KAIST, Daejeon, South Korea
\and
Max Planck Institute for Informatics, Saarbr\"ucken, Germany\\
\email{msokolow@mpi-inf.mpg.de}}

\begin{document}

\maketitle

\begin{abstract}
Let $E$ be a set of $n$ elements, $f\colon2^E\to\mathbb Z_{\ge0}$ be submodular,
$t_1,\ldots,t_r\in E$ be distinct, and $k\in\mathbb N$.
The \emph{Submodular Multiway Partition} problem asks whether $E$ has
a partition $(P_1,\ldots,P_r)$ with $t_i\in P_i$ for all $i$ and
$\sum_i f(P_i)\le k$.  We prove that this problem can be solved in
$r^{O(k)}n^{O(1)}$ time whenever $f$ is posimodular. This covers monotone submodular functions, matroid
rank functions, and nonnegative combinations of monotone and
symmetric submodular functions.  When $f$ is a connectivity function, we remove the dependence on the number of terminals
and obtain a $k^{O(k)}n^{O(1)}$-time algorithm.  The proof is entirely elementary, and is arguably the
simplest proof that Edge Multiway Cut is fixed-parameter
tractable.  As a new consequence, our main result immediately implies that Hypergraph Multiway Partition is
fixed-parameter tractable.

\keywords{Submodular functions \and Multiway partition \and
Parameterized algorithms \and Hypergraph partitioning \and
Posimodularity}
\end{abstract}

\section{Introduction}
\label{sec:introduction}

Let $E$ be a set of $n$ elements, $f\colon2^E\to\mathbb Z_{\ge0}$ be submodular,
 $T=\{t_1,\ldots,t_r\}\subseteq E$ be a set of $r$
\emph{terminals}, and $k\in\mathbb N$ be a \emph{budget}.
A partition $(P_1,\ldots,P_r)$ of $E$ is \emph{feasible} if
$t_i\in P_i$ for every $i$, and its \emph{value} is
$\sum_{i=1}^r f(P_i)$.  We write $\OPT(f,T)$ for the minimum value of
a feasible partition.  The \emph{Submodular Multiway Partition} problem asks whether
$\OPT(f,T)\le k$.

The graph cut function is the motivating example.  Let $G$ be a graph.  For each $X \subseteq V(G)$, let \(f(X)\) be the number of
edges of \(G\) with one end in \(X\) and one end in \(V(G)\setminus X\). It is easy to verify that \(f\) is a
normalized symmetric submodular function, and \(\sum_i f(P_i)\) is twice the number of edges crossing the partition.  Thus, Submodular Multiway Partition contains Edge Multiway Cut after doubling the budget.

Edge Multiway Cut is NP-hard already for three
terminals, as shown by Dahlhaus, Johnson, Papadimitriou, Seymour, and
Yannakakis~\cite{DahlhausJPSY94}, who also gave a
\((2-2/r)\)-approximation.  On the parameterized side,
Marx~\cite{Marx06} proved that both Edge Multiway Cut and Vertex Multiway Cut are fixed-parameter tractable in
the cut size \(k\), introducing important separators, which have
since become a standard tool \cite{Marx11,CyganFKLMPPS15}.
Xiao~\cite{Xiao10} gave an \(O^*(2^k)\) algorithm for
Edge Multiway Cut, later improved to
\(O^*(1.84^k)\)~\cite{CaoCF14}.  For Vertex Multiway Cut,
sharper branching gives an \(O^*(4^k)\) algorithm~\cite{ChenLL09},
later improved to \(O^*(2^k)\) via half-integral
relaxations~\cite{CyganPPW13}.  Flow augmentation is also a powerful new technique that can be used to derive
fixed-parameter tractable algorithms for both versions (see~\cite{KMPSW24}).
However, all of these methods
use graph notions such as paths, reachability, and
deletion-minimal separators, which have no obvious meaning for arbitrary submodular functions.



Our focus is the value-oracle setting.  
 A \emph{connectivity function} on
\(E\) is a function \(f\colon 2^E\to\mathbb Z\) satisfying, for all
\(X,Y\subseteq E\),
\begin{align*}
  f(X)+f(Y)&\ge f(X\cap Y)+f(X\cup Y), & \text{(submodular)}
  \\
  f(X)&=f(E\setminus X), &\text{(symmetric)}
  \\
  f(\varnothing)&=0. &\text{(normalized)}
\end{align*}  
Connectivity
functions arise from graphs, hypergraphs, and matroids, and form a
basic abstraction in combinatorial optimization. 
A function $f\colon2^E\to\mathbb Z$ is \emph{posimodular} if $f(A)+f(B)\ge f(A\setminus B)+f(B\setminus A)$ for all $A,B \subseteq E$.

We prove two main results.  First, when $f$ is posimodular, we give an FPT algorithm for Submodular Multiway Partition parameterized by $k+r$.  Second, when $f$ is a connectivity function, we give an FPT algorithm for Submodular Multiway Partition parameterized by only the budget $k$.

\begin{theorem}
\label{thm:main}
Let $E$ be a set of $n$ elements,  $f\colon2^E\to\mathbb Z_{\ge0}$ be submodular,
 $T=\{t_1,\ldots,t_r\}\subseteq E$,  $k\in\mathbb N$, $\OPT(f,T)$ be the minimum value of
a feasible partition, and $\gamma$ be the time required to query one value of
$f\colon2^E\to\mathbb Z_{\ge0}$.
\begin{enumerate}
\item[(i)] If $f$ is posimodular, then one can decide
in time
\[
  O\!\left(r^{k+1}n^6(\gamma+n)\right)
\]
whether $\OPT(f,T)\le k$ and, if so, return $\OPT(f,T)$ together
with a feasible partition attaining it.
\item[(ii)] If $f$ is a connectivity function, then one can decide
in time
\[
  O\bigl(k^{k+1}n^6(\gamma+n)\bigr)
\]
whether $\OPT(f,T)\le k$ and, if so, return $\OPT(f,T)$ together
with a feasible partition attaining it.
\end{enumerate}
\end{theorem}

Submodular Multiway Partition was previously studied mainly from
the viewpoint of approximation.  Zhao, Nagamochi, and
Ibaraki~\cite{ZhaoNI05} analyzed greedy splitting algorithms, Chekuri
and Ene~\cite{ChekuriE11} obtained a \(2\)-approximation in general
and a \((3/2-1/r)\)-approximation in the symmetric case via the
Lov\'asz extension. Ene, Vondr\'ak, and Wu~\cite{EneVW13}
established matching hardness in the value-oracle model (see
also~\cite{BiCJ25} for monotone objectives).  The terminal-free
problem of partitioning into \(r\) nonempty parts is different in
character. It is solvable in polynomial time for \(r\le4\)
\cite{Queyranne98,HirayamaLMSX24}.  To our knowledge, nothing was known about the
parameterized complexity of Submodular Multiway Partition beyond the graph case.

It is easy to show that every
symmetric submodular function is posimodular.  Moreover, every monotone function is also posimodular.
Submodularity and posimodularity are also preserved under nonnegative linear combinations.  Thus~(i) applies to nonnegative integer-valued monotone submodular
functions, such as matroid rank functions, and to integer-valued
nonnegative combinations of monotone and symmetric submodular
functions. 

We have already noted that the graph cut function is a connectivity function.  Thus,~(ii) immediately implies that Edge Multiway Cut is FPT in the budget $k$.  Our proof is arguably the simplest known proof of this result.  

As a new application, let $H$ be a hypergraph with nonnegative integer
hyperedge weights $w$, and for each $X \subseteq V(H)$ define
\[
  f_H(X):=\sum_{e\in  E(H)}
    w_e\,\1[\,e\cap X\ne\varnothing\ne e\setminus X].
\]
Minimizing $\sum_i f_H(P_i)$ over
feasible partitions is the \emph{Hypergraph Multiway Partition}
problem studied by Zhao, Nagamochi, and Ibaraki and by Chekuri and
Ene~\cite{ZhaoNI05,ChekuriE11}. A hyperedge meeting $q \ge 2$ parts
contributes $qw_e$, whereas a hyperedge contained in a single part contributes $0$.  Since $f_H$ is a connectivity function, Theorem~\ref{thm:main}(ii) immediately yields the following corollary. 

\begin{corollary}
\label{cor:hypergraph}
Hypergraph Multiway Partition with nonnegative integer hyperedge
weights is fixed-parameter tractable parameterized by the objective value.
\end{corollary}


The proof of Theorem~\ref{thm:main} is easy enough to describe in a paragraph.  
For each terminal $t_i$, replace $f$ by its
\emph{terminal closure}
\[
  g_i(X)=\min\{f(Y\cup\{t_i\}):X\subseteq Y\subseteq E\setminus T\},
\]
which is a monotone submodular function.  We assign the nonterminals to
terminals one at a time.  Assigning $v$ to $t_i$ is \emph{free} if it
does not change the closure value of the set currently assigned to
$t_i$. By monotonicity and submodularity, free assignments are always safe.  If
no assignment is free, we branch over all $r$ assignments, each of
which costs at least one unit of budget, so the branching depth is at
most~$k$.  The search partitions the nonterminals, but the minimizing
supersets that realize the values $g_i$ may overlap.  However, by posimodularity, we
can uncross these sets without increasing the total value.  
For~(ii), we use symmetry to decompose $E$
into `connected blocks'.  A block containing
$q\ge2$ terminals contributes at least $q$ to the optimum, so a
yes-instance has at most $k$ terminals in such blocks. Thus,~(i)
can be applied block by block.




\section{Proof of Theorem~\ref{thm:main}(i)}
\label{sec:direct}

The first ingredient is to solve the following covering problem.  Let $g_1,\ldots,g_r\colon2^U\to\mathbb Z_{\ge0}$ be monotone
submodular functions, each evaluable in time $\rho$, and let
$k\in\mathbb N$.  Call $(X_1,\ldots,X_r)$ a \emph{cover} of
$U$ if $\bigcup_iX_i=U$, and let its \emph{value} be
$\sum_ig_i(X_i)$.  

\begin{lemma}
\label{lem:cover}
In time $O(r^{k+1}\abs U(\rho+\abs U))$ one can
decide whether $U$ has a cover of value at most $k$ and, if so, return
a cover of minimum value whose sets are pairwise disjoint.
\end{lemma}

\begin{proof}
We propose the following branching scheme that searches over pairwise disjoint sets $S_1,\ldots,S_r\subseteq U$.  Initially, each $S_i=\varnothing$.  At each node, \emph{reject} if
$\sum_ig_i(S_i)>k$, and \emph{accept} if the $S_i$ cover $U$.
Otherwise pick an uncovered element $v$.  If
$g_i(S_i\cup\{v\})=g_i(S_i)$ for some $i$, add $v$ to $S_i$ and stay
at the same node. We call this a  \emph{free addition}.  If there is no free
addition, create $r$ children, the $i$th of which adds $v$ to $S_i$.
Explore the whole tree and return an accepting node of minimum value.

Every accepting node is a cover of value at most $k$, so a returned
cover is as claimed.  Conversely, suppose that a cover of value at
most $k$ exists, let $\mu$ be the minimum value of a cover, and call a
node \emph{good} if some cover $(X_1,\ldots,X_r)$ of value $\mu$
\emph{extends} its sets, that is, $S_i\subseteq X_i$ for all $i$.  The
root is good.  A free addition of $v$ to $S_i$ preserves goodness,
since monotonicity and submodularity give
\[
  0\le g_i(X_i\cup\{v\})-g_i(X_i)
  \le g_i(S_i\cup\{v\})-g_i(S_i)=0.
\]
Thus, replacing $X_i$ by $X_i\cup\{v\}$ yields again a cover of value
$\mu$.  A good branching node has a good child, since $v\in X_i$ for
some $i$.  A good node is not rejected, because by monotonicity its
value is at most $\mu\le k$.  Hence some good node is accepted, and
its value is at most $\mu$, so exactly $\mu$.  Finally, elements are
added only while uncovered, each to exactly one set, so every accepted
cover is a partition of $U$.

For the running time, note that no addition at a branching node is
free, so by integrality and monotonicity each child increases
$\sum_ig_i(S_i)$ by at least~$1$.  Hence a branching node of depth $d$
has value at least $d$, and since it is not rejected, $d\le k$.  The
children of a branching node of depth $k$ have value at least $k+1$
and are rejected at once (their values were computed at the parent),
so all evaluations occur at nodes of depth at most $k$.  These form an
$r$-ary tree of depth at most $k$, which has at most $r^k$ leaves.
Along a root-to-leaf path each element of $U$ is examined at most
once, at a cost of $r$ evaluations of the $g_i$ and $O(\abs U)$ set
operations, so the total time is
$O(r^k\abs U(r\rho+\abs U))=O(r^{k+1}\abs U(\rho+\abs U))$. \qed
\end{proof}

The second ingredient uses posimodularity to turn covers into partitions.

\begin{lemma}
\label{lem:uncross}
Let $f\colon2^E\to\mathbb Z$ be posimodular with value queries of
cost $\gamma$, and let $A_1,\ldots,A_r$ cover $E$.  In time
$O(r^2(rn+\gamma))$ one can find a partition $(P_1,\ldots,P_r)$ of
$E$ such that, for every $i$,
\[
  A_i\setminus\bigcup_{j\ne i}A_j\subseteq P_i\subseteq A_i
  \qquad\text{and}\qquad
  f(P_i)\le f(A_i).
\]
\end{lemma}

\begin{proof}
Start with $Z_i=A_i$ for all $i$.  While some $Z_i$ and $Z_j$
intersect, posimodularity gives
\[
  f(Z_i)+f(Z_j)\ge f(Z_i\setminus Z_j)+f(Z_j\setminus Z_i).
\]
Thus, either $f(Z_i) \geq f(Z_i\setminus Z_j)$ or $f(Z_j) \geq f(Z_j\setminus Z_i)$.  If the former holds, we replace $Z_i$ by $Z_i\setminus Z_j$.  If the latter holds and the former does not, we replace $Z_j$ by $Z_j\setminus Z_i$.  Each step makes the chosen pair
disjoint. Moreover, since sets only shrink, no new intersecting
pairs are created.  Hence after at most $\binom r2$ steps the $Z_i$ are pairwise
disjoint.  A step finds an intersecting pair in $O(rn)$ time and uses
two value queries and $O(n)$ set operations. \qed
\end{proof}

Define \(\SFM(n,\gamma)\) to be the time required to minimize a submodular
function on an \(n\)-element set and return a minimizer, when one
evaluation of the function takes time \(\gamma\). Orlin's algorithm~\cite{Orlin09} shows that
\(\SFM(n,\gamma) \in O(n^5(\gamma+n))\).

Let $T=\{t_1,\ldots,t_r\}$ be the terminals and $U=E\setminus T$.  Recall that the 
\emph{terminal closure} of $t_i$ is the function
\[
  g_i(X)=\min\{f(Y\cup\{t_i\}):X\subseteq Y\subseteq  U\}.
\]
Note that $g_i(X)$, together with a minimizing superset, is obtained
by minimizing the submodular function $W\mapsto f(W\cup X\cup\{t_i\})$
over $W\subseteq U\setminus X$. An evaluation costs one value query
and $O(n)$ set operations, so this takes time $\SFM(n,\gamma+n)$.

The final ingredients are the following properties of the terminal closure.
\begin{lemma}
\label{lem:closures}
Let $f\colon2^E\to\mathbb Z_{\ge0}$ be submodular with value queries
of cost $\gamma$.
\begin{enumerate}
\item[(a)] Every $g_i$ is nonnegative, integer-valued, monotone, and
submodular.
\item[(b)] If $(P_1,\ldots,P_r)$ is a feasible partition, then
$(P_i\setminus\{t_i\})_{i=1}^r$ is a cover of $U$ of value at most
$\sum_if(P_i)$.
\item[(c)] If $f$ is also posimodular, then from any cover
$(X_1,\ldots,X_r)$ of $U$ one can compute, in time
$O(r\,\SFM(n,\gamma+n)+r^2(rn+\gamma))=O(rn^5(\gamma+n))$, a
feasible partition of value at most $\sum_ig_i(X_i)$.
\end{enumerate}
In particular, if $f$ is also posimodular, then
$\OPT(f,T)$ equals the minimum value of a cover of $U$.
\end{lemma}

\begin{proof}
(a) Nonnegativity and integrality are inherited from $f$.  Moreover, $g_i$ is clearly
monotone.  For submodularity, let $A'\supseteq A$ and $B'\supseteq B$
attain $g_i(A)$ and $g_i(B)$.  Since $A'\cup B'\supseteq A\cup B$ and
$A'\cap B'\supseteq A\cap B$ are admissible in the respective minima,
submodularity of $f$ gives
\[
\begin{aligned}
g_i(A)+g_i(B)
 &=f(A'\cup\{t_i\})+f(B'\cup\{t_i\})\\
 &\ge f(A'\cup B'\cup\{t_i\})+f((A'\cap B')\cup\{t_i\})\\
 &\ge g_i(A\cup B)+g_i(A\cap B).
\end{aligned}
\]

(b) No $P_i$ contains a terminal other than $t_i$, so the sets
$X_i=P_i\setminus\{t_i\}$ cover $U$, and
$g_i(X_i)\le f(X_i\cup\{t_i\})=f(P_i)$.

(c) For each $i$ choose $Z_i$ with $X_i\subseteq Z_i\subseteq U$ and
$f(Z_i\cup\{t_i\})=g_i(X_i)$, and put $A_i=Z_i\cup\{t_i\}$.  The
$A_i$ cover $E$, and $t_i$ lies in $A_i$ and in no other $A_j$.
Lemma~\ref{lem:uncross} therefore returns a partition
$(P_1,\ldots,P_r)$ with $t_i\in P_i$ and
$f(P_i)\le f(A_i)=g_i(X_i)$ for all~$i$.  Each $Z_i$ is found by one
submodular minimization, and Lemma~\ref{lem:uncross} takes
$O(r^2(rn+\gamma))$ time. Since $r\le n$, this is
$O(rn^5(\gamma+n))$. \qed
\end{proof}

Note that the disjointness guaranteed by Lemma~\ref{lem:cover} does
not make posimodularity dispensable. The sets $X_i$ are disjoint, but
$f(X_i\cup\{t_i\})$ may exceed $g_i(X_i)$, and the minimizing
supersets $Z_i$ that realize the values $g_i(X_i)$ may overlap.
Posimodularity is precisely what allows them to be uncrossed without
loss.

We are now ready to prove Theorem~\ref{thm:main}(i).

\begin{proof}
Run Lemma~\ref{lem:cover} on $g_1,\ldots,g_r$ with
$\rho=\SFM(n,\gamma+n)$.  By Lemma~\ref{lem:closures} the answer is
correct, and a returned minimum-value cover is converted by
Lemma~\ref{lem:closures}(c) into a feasible partition of value at most
its value, hence of value $\OPT(f,T)$.  The conversion costs
$O(rn^5(\gamma+n))$, which is dominated.  Hence, the total runtime is
$O\left(r^{k+1}n^6(\gamma+n)\right)$. \qed
\end{proof}

\section{Proof of Theorem~\ref{thm:main}(ii)}
\label{sec:components}

Throughout this section $f\colon 2^E \to \mathbb{Z}$ is a connectivity function.  Call $f$
\emph{connected} if $f(X)>0$ for every nonempty proper subset
$X\subset E$.  We will show how to decompose $E$ into `connected blocks'.  These results are well known and easy, but for completeness, we include the proofs. For $A\subseteq E$, let $\co A=E \setminus A$ and write $f|_A$ for the restriction
of $f$ to $2^A$.

\begin{lemma}
\label{lem:zero-split}
If $f(A)=0$, then for all $X \subseteq E$
\begin{equation}
  f(X)=f(X\cap A)+f(X\setminus A),
  \label{eq:zero-additivity}
\end{equation}
and $f|_A$ and $f|_{\co A}$ are connectivity functions.
\end{lemma}

\begin{proof}
Since connectivity functions are subadditive, $f(X\cap A)+f(X\setminus A)\ge f(X)$.  For the reverse
inequality, observe
\[
\begin{aligned}
f(X)
 &=f(X)+f(A)
  \ge f(X\cap A)+f(X\cup A)
  =f(X\cap A)+f(\co X\cap\co A)+f(A)\\
 &\ge f(X\cap A)+f(\co X\cup A)
  =f(X\cap A)+f(X\setminus A).
\end{aligned}
\]
Restrictions of normalized submodular functions are normalized and
submodular.  For $X\subseteq A$, symmetry of $f$
and~\eqref{eq:zero-additivity} applied to $\co X$ yield
\[
  f(X)=f(\co X)=f(\co X\cap A)+f(\co X\setminus A)
  =f(A\setminus X)+f(\co A)=f(A\setminus X),
\]
since $f(\co A)=f(A)=0$. Thus, $f|_A$ is symmetric.  The same argument with
$\co A$ in place of $A$ shows that $f|_{\co A}$ is also symmetric. \qed
\end{proof}

If $f$ is not connected, Lemma~\ref{lem:zero-split} splits it along a
nonempty proper set of value zero into two connectivity functions, to
which the lemma applies again.  By induction, $E$ has a partition into
\emph{blocks} $E_1,\ldots,E_m$ such that each $f_j=f|_{E_j}$ is a
connected connectivity function and for all $X \subseteq E$,
\begin{equation}
  f(X)=\sum_{j=1}^m f_j(X\cap E_j).
  \label{eq:block-additivity}
\end{equation}

Let $T=\{t_1,\ldots,t_r\}\subseteq E$ be the set of terminals, and let $T_j=T\cap E_j$ and $q_j=\abs{T_j}$. If $q_j\ge1$ let
$\OPT_j=\OPT(f_j,T_j)$ be the optimum value of the restricted instance on the
block $E_j$. The following lemma shows how to combine optimal partitions of each block into an optimal partition of $E$.

\begin{lemma}
\label{lem:opt-decomposes}
\[
  \OPT(f,T)=\sum_{j:\,q_j\ge1}\OPT_j.
\]
\end{lemma}

\begin{proof}
Take an optimum feasible partition of each block with
$q_j\ge1$ and assign each block with $q_j=0$ wholly to an arbitrary
part.  The parts are indexed by distinct terminals, so this is a
feasible partition of $E$. By~\eqref{eq:block-additivity} and
$f_j(E_j)=0$ its value is $\sum_{j:\,q_j\ge1}\OPT_j$.

For the other inequality, let $(P_t)_{t\in T}$ be a feasible partition of $E$.
By~\eqref{eq:block-additivity}, its value is the sum over $j$ of the
contributions $\sum_{t\in T}f_j(P_t\cap E_j)$.  Fix $j$ with $q_j\ge1$, pick
$t_0\in T_j$, and merge every piece $P_t\cap E_j$ with $t\notin E_j$
into $P_{t_0}\cap E_j$.  Since connectivity functions are subadditive, this does not increase
$\sum_tf_j(P_t\cap E_j)$, and the result is a feasible partition of
$E_j$ with terminal set $T_j$.  Hence $E_j$ contributes at least
$\OPT_j$, while a block with $q_j=0$ contributes at least $0$ (since connectivity functions are nonnegative). \qed
\end{proof}

We are now ready to prove Theorem~\ref{thm:main}(ii).

\begin{proof}
Queyranne's algorithm~\cite{Queyranne98} either finds a nonempty proper
set of value zero or certifies connectedness.  Thus, the block decomposition can be computed in $O(n^4(\gamma+n))$ time. A block with $q_j=1$ has
$\OPT_j=f_j(E_j)=0$.  In a block with $q_j\ge2$, every part of a
feasible partition contains one terminal and misses another, so it
is a nonempty proper subset of $E_j$ and has value at least $1$ by
connectedness. Hence $\OPT_j\ge q_j$.  Reject if
$\sum_{j:\,q_j\ge2}q_j>k$.  Otherwise every block with $q_j\ge2$ has
at most $k$ terminals, and there are at most $k/2$ such blocks.
Apply part~(i) to each of them with budget $k$. Reject if some
$\OPT_j$ exceeds $k$. Otherwise accept if and only if
$\sum_{j:\,q_j\ge2}\OPT_j\le k$, combining the returned partitions
via Lemma~\ref{lem:opt-decomposes}.

Let $E_j$ be a block with $2 \le q_j\le k$ terminals. By Theorem~\ref{thm:main}(i), we can compute an optimal partition of $E_j$ in time
$O(k^{k+1}\abs{E_j}^6(\gamma+\abs{E_j}))$.  Summing over all such blocks and adding the decomposition time
gives $O(k^{k+1}n^6(\gamma+n))$. \qed
\end{proof}

\subsubsection*{Acknowledgments}
Tony Huynh and Sang-il Oum were supported by the Institute for Basic
Science (IBS-R029-C1). Roohani Sharma was supported by the Institute for
Basic Science (IBS-R029-Y8).

\subsubsection*{AI Declaration}
All results and proofs in this paper were found by humans.  AI tools were used for language editing and typesetting assistance.    The authors take full
responsibility for the mathematical content and the final manuscript.

\begingroup
\emergencystretch=2em
\bibliographystyle{splncs04}
\bibliography{ipco2027}
\endgroup

\end{document}